\documentclass[runningheads]{llncs}
\usepackage[utf8]{inputenc}
\usepackage{mathtools}
\usepackage{amsfonts, amssymb, amsbsy}
\usepackage{subfigure}
\usepackage{color}
\usepackage{hyperref}
\usepackage{cite}
\usepackage[font=small]{caption}
\usepackage[basic]{complexity}
\usepackage{times}

\newlength{\toplength}
\newlength{\bottomlength}
\title{Simultaneous Embeddability of Two Partitions}

\author{Jan Christoph Athenstädt\inst{1} \and Tanja Hartmann\inst{2} \and Martin Nöllenburg\inst{2}}

\institute{Department of Computer and Information Science, University of Konstanz, Germany \and Institute of Theoretical Informatics, Karlsruhe Institute of Technology (KIT), Germany}

\begin{document}
\maketitle
\begin{abstract}
  We study the simultaneous embeddability of a pair of partitions of
  the same underlying set into disjoint blocks. Each element of the set is mapped
  to a point in the plane and each block of either of the two
  partitions is mapped to a region that contains exactly those points that belong to the elements in
  the block and that is bounded by a simple closed curve.  We establish three main classes of simultaneous
  embeddability (\emph{weak}, \emph{strong}, and \emph{full}
  embeddability) that differ by increasingly strict well-formedness
  conditions on how different block regions are allowed to intersect.
  We show that these simultaneous embeddability classes are closely
  related to different planarity concepts of hypergraphs.  For each
  embeddability class we give a full
  characterization. %
  We show that (i) every pair of partitions has a weak simultaneous
  embedding, (ii) it is \NP-complete to decide the existence of a
  strong simultaneous embedding, and (iii) the existence of a full
  simultaneous embedding can be tested in linear time.
\end{abstract}

\section{Introduction}

Pairs of partitions of a given set of objects occur naturally when 
evaluating two alternative clusterings in the field of data analysis and data mining.
A \emph{clustering} partitions
a set of objects into \emph{blocks} or \emph{clusters}, such that
objects in the same cluster are more similar (according
to some notion of similarity) than objects in different clusters.
There are a multitude of clustering algorithms that use, e.g., an underlying graph structure or
an attribute-based distance measure to define similarities.
Many
algorithms also provide configurable parameter settings.
Consequently, different algorithms return different
clusterings and judging which clustering is the most meaningful with
respect to a certain interpretation of the data must be done by a
human expert.
For a structural comparison of two clusterings several
numeric measures exist~\cite{ww-ccao-07}, however, a single numeric value
hardly shows where the clusterings agree or disagree.  Hence, a data
analyst may want to 
compare different clusterings
visually, which motivates the study of simultaneous embeddability of
two partitions.  %

We provide  fundamental characterizations and
complexity results regarding the simultaneous embeddability of a pair of
partitions.  While simultaneous embeddability can generally be
defined for any number $k\ge 2$ of partitions, we
focus on the basic case of embedding \emph{two} partitions, which is
also the most relevant one in the data analysis application.
We propose to
embed two alternative partitions of the same set $U$
into the plane $\mathbb R^2$ by mapping each element of $U$ to a unique point and
each block (of either of the two partitions) to a region bounded by a
simple closed curve. Each block region must contain all points that
belong to elements in that block and no point whose element belongs to a
different block. 
Hence, in total, each point lies inside two block regions.

A simultaneous embedding of two partitions shares certain properties
with set visualizations like Euler or Venn
diagrams~\cite{c-gdaevd-07,ffh-edg-08,saa-favos-09}.  Its readability
will be affected by well-formedness
conditions for the intersections of the different block regions.
Accordingly, we define a (strict) hierarchy of embeddability classes based on
increasingly 
tight well-formedness conditions: \emph{weak}, \emph{strong}, and \emph{full} embeddability.
We show that (i) any two partitions
are weakly embeddable, (ii) the decision problem for strong embeddability is \NP-complete, and (iii) there is a linear-time decision algorithm for full
embeddability. %
We fully
characterize the embeddability classes in terms of the existence of a
planar support
(strong embeddability) or in terms of the
planarity of the %
bipartite map 
(full
embeddability). Interestingly, both concepts are closely related to
hypergraph embeddings and different notions of hypergraph planarity.
Our \NP-completeness result %
implies that
vertex-planarity testing of 2-regular 
hypergraphs is also \NP-complete.

\subsection{Related Work}

In information visualization there are a large variety of techniques for visualizing clusters of objects, some of which simply map objects to (colored) points so that spatial proximity indicates object similarity~\cite{bsldhc-dvwms-08,k-sm-01}, others explicitly visualize clusters or general sets as regions in the plane~\cite{cpc-bsrrwioev-09,saa-favos-09}. These approaches are visually similar to Euler diagrams~\cite{c-gdaevd-07,ffh-edg-08}, however, they do not give hard guarantees on the final set layout, e.g., in terms of intersection regions or connectedness of regions, nor do they specifically consider the simultaneous embedding of two or more clusterings or partitions.

\emph{Clustered planarity} is a concept in graph drawing that combines a planar graph layout with a drawing of the clusters of a single hierarchical clustering. Clusters are represented as regions bounded by simple closed and pairwise
crossing-free curves. Such a layout is called \emph{c-planar} if no
edge crosses a region boundary more than once~\cite{fce-pcg-95}.  

The simultaneous embedding of two planar graphs on the same vertex set is a topic that is well studied in the graph drawing literature, see the recent survey of Bläsius et al.~\cite{bkr-sepg-13}. In a simultaneous graph embedding each vertex is located at a unique position and edges contained in both graphs are represented by the same curve for both graphs. The remaining (non-shared) edges are embedded so that each graph layout by itself is crossing-free, but edges from the first graph may cross edges in the second graph. %

Some of our results and concepts in this paper can be seen as a generalization of simultaneous graph embedding to  simultaneous hypergraph embedding if we consider blocks as hyperedges: all vertices are mapped to unique points in the plane and two hyperedges, represented as regions bounded by simple closed curves, may only intersect if they belong to different hypergraphs or if they share common vertices. Several concepts for visualizing a single hypergraph are known~\cite{jp-hpcdvd-87,m-dh-90,kks-sdh-09,bkms-psh-10,bcps-psh-12}, 
but to the best of our knowledge the simultaneous layout of two or more hypergraphs has not been studied.

\subsection{Preliminaries}\label{sec:prelim}
Let $U = \{u_1, \dots, u_m\}$ be a finite universe.  A
\emph{partition} $\mathcal P = \{B_1, \dots, B_n\}$ of~$U$ groups the
elements of $U$ into disjoint \emph{blocks}, i.e., every element $u
\in U$ is contained in exactly one block $B_i \in \mathcal P$.  In
this paper, we consider pairs $\{\mathcal{P}_0, \mathcal{P}_1\}$ of
partitions of the same universe $U$, i.e., each element $u \in U$ is
contained in one block of $\mathcal{P}_0$ and in one block of
$\mathcal{P}_1$. In the following
we often omit to mention~$U$ explicitly.

Let $\mathcal{S}$ be a collection of subsets of~$U$. An
\emph{embedding} $\Gamma$ of $\mathcal{S}$ maps every element \mbox{$u \in
U$} to a distinct point $\Gamma(u) \in \mathbb R^2$ and every set $S
\in \mathcal{S}$ to a simple, bounded, and closed region $\Gamma(S)
\subset \mathbb R ^2$ such that $\Gamma(u) \in \Gamma(S)$ if and only
if $u \in S$. Moreover, we require that each contiguous intersection
between the boundaries of two regions is in fact a \emph{crossing
  point} $p \in \mathbb R^2$, i.e., the local cyclic order of the
boundaries alternates around $p$.  A \emph{simultaneous embedding}
$\Gamma$ of a pair of partitions $\{\mathcal{P}_0, \mathcal{P}_1\}$ is
an embedding of the union $\mathcal{P}_0 \cup \mathcal{P}_1$ of the
two partitions. We define $R_B = \Gamma(B)$ as the \emph{block region}
of a block $B$ and denote its boundary by $\partial
R_B$. Figure~\ref{fig:embedding-examples} shows examples of
simultaneous embeddings in the three different embedding classes to be
defined in Section~\ref{s:mce}.

\begin{figure}[tb]
	\centering
		\subfigure[\label{sfg:weak}weak embedding]{\includegraphics[page=1,width=.275\textwidth]{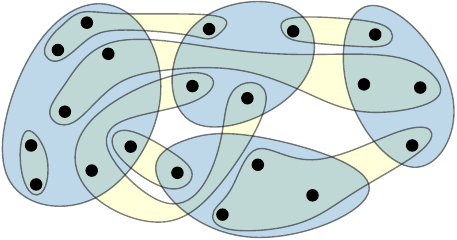}}
		\hfill
		\subfigure[\label{sfg:strong}strong embedding]{\includegraphics[page=2,width=.275\textwidth]{embeddings}}
		\hfill
		\subfigure[\label{sfg:full}full embedding]{\includegraphics[page=3,width=.275\textwidth]{embeddings}}
	\caption{Examples of simultaneous embeddings of two partitions.}
	\label{fig:embedding-examples}
\end{figure}

A simultaneous embedding $\Gamma$ induces a subdivision of the plane
and we can derive a plane multigraph $G_\Gamma$ by introducing a node
for each intersection of two boundaries and an edge for each section
of a boundary that lies between two intersections.  Furthermore, a
boundary without intersections is replaced by a node with a self loop
nested inside its surrounding face.  We call $G_\Gamma$ the
\emph{contour graph} of $\Gamma$ and its dual graph $G_\Gamma^*$ the
\emph{dual graph} of $\Gamma$.  The faces of~$G_\Gamma$ belong to
zero, one, or two block regions.  We call a face that belongs to no
block region a \emph{background face}, a face that belongs to a single
block region a \emph{linking face}, and a face that belongs to two
block regions an \emph{intersection face}.  Only intersection faces
contain points corresponding to elements in the universe, and no two
faces of the same type are adjacent in the contour graph.

Alternatively, the union of the two partitions $\mathcal{P}_0 \cup
\mathcal{P}_1$ can also be seen as a
hypergraph $H=(U,\mathcal{P}_0 \cup \mathcal{P}_1)$, where every
element~$u \in U$ is a vertex and every block defines a
hyperedge, i.e., a non-empty subset of $U$. The hypergraph $H$ is \emph{2-regular} since every vertex is
contained in exactly two hyperedges. %
We denote $H = H(\mathcal{P}_0, \mathcal{P}_1)$ as the
\emph{corresponding hypergraph} of the pair of partitions
$\{\mathcal{P}_0, \mathcal{P}_1\}$.

\emph{Hypergraph supports}~\cite{kks-sdh-09} play an important role in
hypergraph embeddings and their planarity.  A support of a hypergraph
$H = (V,\mathcal{S})$ is a graph $G_p = (V, E)$ on the vertices of
$H$, such that the \emph{induced subgraph} $G_p[S]$ of every hyperedge
$S \in \mathcal{S}$ is connected.  We extend the concept of supports
to pairs of partitions, i.e., we say that a graph $G_p = (V, E)$ is a
support for $\{\mathcal{P}_0, \mathcal{P}_1\}$, if it is a support of
$H(\mathcal{P}_0, \mathcal{P}_1)$.

We call a support \emph{path based}, if the induced subgraphs of all
hyperedges are paths,\footnote{Brandes et al.~\cite{bcps-psh-12} used
  a slightly different definition and called a support \emph{path
    based} if the induced subgraph of each hyperedge has a Hamiltonian
  path.}  and \emph{tree based}, if all hyperedge-induced
subgraphs are trees, i.e., they do not contain any cycles. For any support
$G_p$ of a pair of partitions~$\{\mathcal P_0,\mathcal P_1\}$ we can
always create a tree-based support $G'_p$ by removing edges from cycles:
Suppose there exists a block $B \in \mathcal{P}_0$ such that
$G_p[B]$ contains a cycle~$K$.
If the vertices in~$K$ are also contained in a common block of
$\mathcal{P}_1$,
we can just remove a random edge from $K$ without
destroying the support property. 
Otherwise, we can remove an
edge from $K$ that connects vertices in two different blocks of
$\mathcal{P}_1$ without destroying the support property.

The \emph{bipartite map} $G_b(H)$ of a hypergraph $H = (V,\mathcal S)$
is defined as the bipartite graph~$G_b(H)= (V\cup \mathcal S,E_b)$
that has a node for each vertex in~$V$ and for each hyperedge
in~$\mathcal S$~\cite{w-hvbm-75}.  A node~$v\in V$ is adjacent to a
node~$S\in \mathcal S$ if~$v\in S$. We say that~$G_b(H)$ is the
bipartite map of a pair of partitions $\{\mathcal{P}_0,
\mathcal{P}_1\}$ if~$H = H(\mathcal{P}_0, \mathcal{P}_1)$.

Finally, we define the \emph{block intersection graph} $G_s(\mathcal P_0, \mathcal P_1)$ as the
graph with vertex set $V_s = \mathcal P_0 \cup \mathcal P_1$ and edge
set $E_s = \{\{B,B'\} \mid B \cap B' \ne \emptyset\}$. Thus $G_s$ has
a vertex for each block and an edge between any two blocks that share
a common element.  Since only blocks of different partitions can
intersect, we know that $G_s$ is bipartite.

\section{The Main Classes of Embeddability}\label{s:mce}

We define three main concepts of simultaneous embeddability for pairs of
partitions.  We will see that these
concepts induce a hierarchy of embeddability classes of pairs of
partitions. 

\subsection{Weak Embeddability}
\label{sub:weak_embeddability}

We begin with \emph{weak embeddability}, which is the
most general concept.  
\begin{definition}[Weak Embeddability]\label{def:weak}
  A simultaneous embedding of two partitions is \emph{weak} if
  no two block regions of the same partition intersect.  Two
  partitions are \emph{weakly embeddable} if they have a weak
  simultaneous embedding.
\end{definition}
Prohibiting intersections of block regions of the same
partition is our first well-formed\-ness condition.
A weak embedding emphasizes the fact that
the blocks in each partition are disjoint. %
Since the blocks of any partition are disjoint by definition, it is not surprising that any
pair of partitions is weakly embeddable (see Fig.~\ref{sfg:weak} for an example).

\begin{theorem}\label{thm:weakly_embeddable}
  Any two partitions of a common universe are weakly embeddable on any point set.%
\end{theorem}

\begin{proof}
A spanning forest (in fact, any planar graph) on~$n$ nodes can always be drawn in a planar way on any fixed set of~$n$ points in the plane~\cite{pw-epgfvl-01}. Let now $\mathcal{P}$ be a partition. We choose arbitrary, but distinct points in the plane for the elements of $U$. We then generate a spanning tree on the elements in each block and embed the resulting forest in a planar way on the points. Slightly inflating the thickness of the edges of the trees yields simple bounded block regions. We can do this independently for a second partition on the same points and obtain a weak simultaneous embedding.\qed
\end{proof}
Although the concept of weak embedding does not seem to provide
interesting insights into the structure of a given pair of
partitions, it guarantees at least the existence of a simultaneous embedding for any
pair of partitions that is more meaningful than an arbitrary
embedding.  An obvious drawback of weak embeddings is that the block
regions of disjoint blocks are allowed to intersect, as long as both
blocks belong to different partitions---even if they do not share
common elements.  

\subsection{Strong Embeddability} 
\label{sub:strong_embeddability}

Following the general idea of Euler
diagrams~\cite{c-gdaevd-07}, which do not show regions corresponding
to empty intersections, we establish a stricter concept of
embeddability.  In a strong embedding block regions may only
intersect if the corresponding blocks have at least one element in
common, and even more, each intersection face of the contour graph
must actually contain a point, see Fig.~\ref{sfg:strong}. This is our second
well-formedness condition.
\begin{definition}[Strong Embeddability]\label{def:strong}
  A simultaneous embedding $\Gamma$ of two partitions is
  \emph{strong} if each intersection face of the corresponding contour
  graph contains a point $\Gamma(u)$ for some $u \in U$.  Two
  partitions are \emph{strongly embeddable} if they have a strong
  simultaneous embedding.
\end{definition}
Obviously, a strong embedding is also weak, since blocks of the same
partition have no common elements, and thus, cannot form
intersection faces.
The class of strongly embeddable pairs of partitions is characterized by Theorem~\ref{thm:se_vertex_planar}; we show in Section~\ref{s:cdse} that deciding the strong embeddability of a pair of partitions is \NP-complete.
\begin{theorem}\label{thm:se_vertex_planar}
A pair of partitions of a common universe is strongly embeddable if
and only if it has a planar support.
\end{theorem}
\begin{proof}
  Let~$\{\mathcal P_0,\mathcal P_1\}$ be a pair of partitions
  and let~$G_\Gamma$ be the contour graph resulting from a strong
  embedding~$\Gamma$ of $\{\mathcal P_0,\mathcal P_1\}$.  We construct
  a planar support of~$\{\mathcal P_0,\mathcal P_1\}$ along~$G_\Gamma$
  as follows.
  First recall that the elements of the universe, which correspond to
  the nodes in a support, are represented in~$\Gamma$ by points that
  are drawn inside intersection faces.  Vice versa,
  since~$\Gamma$ is strong, each intersection face contains at least
  one point.  Hence, we choose one point in each intersection face as
  the \emph{center} of this face.
  We now create a dummy vertex for each linking face (observe that one block region may induce several linking faces) and link it to the centers of all adjacent intersection faces. The resulting graph is a subgraph of the dual graph of the contour graph~$G_\Gamma$ and therefore planar. We now connect all remaining vertices in a star-like fashion to the center of their intersection face, routing the edges in a non-crossing way.
  We finally remove the dummy vertices by merging them to an adjacent center, linking all adjacent vertices to that center. This graph remains planar. It also has the support property, since all intersection and linking faces of any block region are connected into a single component, and with them all vertices of that block region.

  Now we construct a strong embedding from a planarly
  embedded support of~$\{\mathcal P_0,\mathcal P_1\}$.  To this end,
  we first construct a tree-based support by deleting edges from
  cycles as described in Section~\ref{sec:prelim}.  Then, we simply
  inflate the edges of each block-induced subtree.  Since the
  underlying support is embedded in a planar way, this yields a simple
  block region for every block in $\{\mathcal P_0,\mathcal P_1\}$
  such that two block regions only intersect at the positions of the
  nodes. Hence, the constructed block regions together with the
  nodes of the support form a strong embedding of~$\{\mathcal
  P_0,\mathcal P_1\}$.  We note that the support graph as a planar graph can in fact be embedded on any point set~\cite{pw-epgfvl-01}. Hence, a strongly embeddable pair of partitions can be strongly embedded on any point set.%
\qed
 \end{proof}

 \subsection{Full Embeddability}
 \label{sub:full_embeddability}
 
 In a strong embedding, a single block region may still cross other block regions and intersect the same block regions several times forming distinct intersection faces---as long as each intersection face contains at least one common point.
 The last of our three embeddability classes prevents this behavior and requires that the block regions form a collection of \emph{pseudo-disks}, i.e., the boundaries of every pair of regions intersect at most twice and the boundaries of two nested regions do not intersect.
 See Fig~\ref{sfg:full} for an example.
 This implies in particular that every block intersection is connected, which is a well-formedness condition widely used in the context of Euler diagrams~\cite{c-gdaevd-07}, and that block regions do not cross and are thus more locally confined.
 \newcommand{\defFullEmbed}{A simultaneous embedding of two partitions is
   \emph{full} if it is strong and additionally the following holds.
   Two block regions have at most one intersection face in common in the
   corresponding contour graph. Furthermore, if an
   intersection face results from two nested blocks, it is adjacent
   to only one linking face, and otherwise, it is adjacent to exactly
   two linking faces.  Two partitions are fully embeddable if they
   have a full simultaneous embedding.}
\renewcommand{\defFullEmbed}{A simultaneous embedding of two partitions is
   \emph{full} if it is a strong embedding and the regions form a collection of pseudo-disks. Two partitions are fully embeddable if they
   have a full simultaneous embedding.}
\begin{definition}[Full Embeddability]\label{def:full}
  \defFullEmbed
\end{definition}
Using a linear-time algorithm for planarity testing~\cite{ht-ept-74},
the following characterization of fully embeddable pairs of
partitions directly implies a linear-time algorithm for deciding full embeddability.
\newcommand{\lemFullEmbedChar}{A pair of partitions of a common
  universe is fully embeddable if and only if its bipartite map is
  planar.}

\begin{theorem}\label{thm:charfull}
  \lemFullEmbedChar
\end{theorem}

\begin{figure}[tb]
  \centering \subfigure[\label{sfg:clus}full
  embedding]{\includegraphics[page=3,width=.3\textwidth]{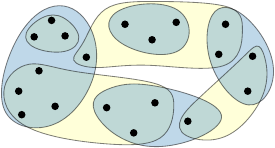}}
  \hfill \subfigure[\label{sfg:dual} dual of contour
  graph]{\includegraphics[page=4,width=.3\textwidth]{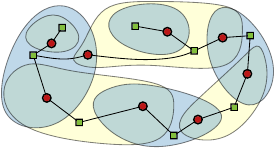}}
  \hfill \subfigure[\label{sfg:bipart}final bipartite
  map]{\includegraphics[page=5,width=.3\textwidth]{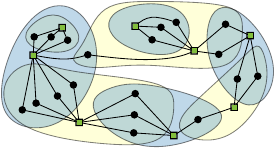}}
  \caption{Step-by-step illustration of the proof of
    Theorem~\ref{thm:charfull}.  Middle: Dual graph of contour graph,
  nodes in background faces already deleted, additional green nodes for
  nested regions already inserted.}
	\label{fig:embedding-examples-full}
\end{figure}

\begin{proof} Let~$\{\mathcal{P}_0, \mathcal{P}_1\}$ be a pair of
  partitions and~$G_\Gamma^*$ the dual graph of a corresponding full
  embedding~$\Gamma$.  The bipartite map of $\{\mathcal{P}_0,
  \mathcal{P}_1\}$ can be constructed as follows.  Remove all vertices
  (and incident edges) from~$G_\Gamma^*$ that stem from a background
  face in the contour graph (Fig.~\ref{sfg:clus} and
  Fig.~\ref{sfg:dual}).  This results in a planar graph with a set of
  \emph{red} nodes resulting from intersection faces, a set of \emph{green} nodes
  resulting from linking faces and edges indicating that two faces in
  the contour graph are adjacent.  Our definitions of simultaneous
  embedding and contour graph ensures that nodes of the same color are
  not adjacent. Hence, the graph constructed so far is
  bipartite. Moreover, each red node is adjacent to at most two green
  nodes and, since~$\Gamma$ is full (more precisely, since each
  intersection face in the contour graph is adjacent to at most two
  linking faces), two green nodes have at most one red common neighbor.
  From the same fact it further follows that each block region
  in~$\Gamma$ induces at most one linking face in the contour graph,
  and thus, the number of green nodes is at most the number of blocks
  in~$\{\mathcal P_0\cup \mathcal P_1\}$.

  If there are fewer green nodes than blocks in~$\{\mathcal P_0\cup
  \mathcal P_1\}$, at least one block region in~$\Gamma$ must be
  completely contained in another block region resulting in a red
  node for the intersection but no green node which could be
  considered as a representative of the nested block.  Hence, the
  red node representing the intersection is only adjacent to one green
  node, namely the green node resulting from the linking face of the
  block region that contains the nested block.  Note that no three
  block regions can be nested, since each point in~$\Gamma$ is
  contained in exactly two block regions.  In this case, we link an
  additional green node to the red node of the nested block region
  such that in the end each block in~$\{\mathcal{P}_0,
  \mathcal{P}_1\}$ corresponds to a green node and each red node is
  adjacent to exactly two green nodes.  Such an additional leaf
  obviously preserves planarity (Fig.~\ref{sfg:dual}).

  Since the bipartite map of a hypergraph (besides the nodes
  representing blocks) consists of nodes representing the elements
  in the universe, we finally replace each red node by a set of nodes
  representing the elements that have been mapped by~$\Gamma$ into the
  corresponding intersection face, and connect each
  of these new nodes along the previous edges to the two green nodes
  previously adjacent to the replaced red node.  This again preserves
  planarity of the finally resulting bipartite map of $\{\mathcal P_0\cup
  \mathcal P_1\}$ (Fig.~\ref{sfg:bipart}).

  Now assume a planar embedding of the bipartite map of~$\{\mathcal
  P_0,\mathcal P_1\}$ with green nodes representing the blocks and
  black nodes representing the elements of the universe.  In order to
  construct a full embedding of~$\{\mathcal P_0,\mathcal P_1\}$, we
  first construct for each block in $\{\mathcal P_0,\mathcal P_1\}$
  a subgraph of the bipartite map such that each subgraph contains the
  green node that represents the corresponding block and two
  subgraphs share exactly one black node if and only if the
  corresponding blocks share at least one element.  Since the map
  of~$\{\mathcal P_0,\mathcal P_1\}$ is bipartite, each of these
  subgraphs is a star with black leaves linked to a green center.
  Together these stars form a planar subgraph of the bipartite map
  such that slightly inflating the edges in each star yields simple
  block regions that intersect exactly at the positions of the black
  nodes.  In the resulting contour graph, two block regions thus
  intersect at most once and each intersection face is adjacent to
  exactly two linking faces, each representing a block.  A nested
  block in $\mathcal P_0\cup P_1$ results in a star that only
  consists of a green center and one black leaf.  In order to completely
  satisfy the condition of a full embedding, we shrink the
  block regions of nested blocks such that the boundary of the inner block 
  does not intersect the boundary of the outer block.  Deleting
  the green and black nodes and drawing a set of points that represent
  the common elements of the intersecting blocks in each
  intersection face finally yields a full embedding.
  
  We note that our construction uses only a single representative element per block intersection. Thus, in contrast to weak and strong embeddings, it is not clear whether a fully embeddable pair of partitions permits a full embedding on any set of points.%
\qed
\end{proof}

\subsection{Hierarchy of Embeddability Classes}
\label{sub:hierarchy_of_embeddability_classes}

A full embedding is strong by definition and we have seen above that a strong
embedding is also weak.  Hence, the three
embeddability classes introduced in this section induce a
hierarchy of embeddability classes. We now show that this hierarchy is strict. 

The left side of Fig.~\ref{sfg:strongnotfull} shows a strong embedding
of a pair of partitions~$\{\mathcal P_0, \mathcal P_1\}$ that is not
fully embeddable.  The dotted lines indicate a planar support proving
the strong embeddability of $\{\mathcal P_0, \mathcal P_1\}$.  The
fact that~$\{\mathcal P_0, \mathcal P_1\}$ is not fully embeddable can
be seen by considering the bipartite map of~$\{\mathcal P_0, \mathcal
P_1\}$, which is a subdivision of~$K_{3,3}$, and thus, is not planar
(see right side of~Fig.~\ref{sfg:strongnotfull}).  The claim then
follows from Theorem~\ref{thm:charfull}.

\begin{figure}[htb]
  \centering \subfigure[\label{sfg:strongnotfull}strongly but not
  fully
  embeddable]{\includegraphics[page=1,width=.45\textwidth]{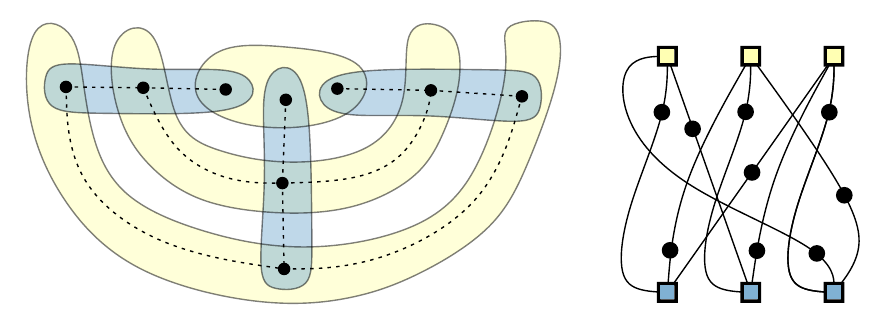}}
  \hspace{2ex} \subfigure[\label{sfg:weaknotstrong}weakly but not
  strongly
  embedable]{\includegraphics[page=3,width=.45\textwidth]{counterex.pdf}}
                \caption{Examples of simultaneous embeddings of two
                  partitions proving the strictness of the
                  hierarchy of embedding classes. }
	\label{fig:counterex}
\end{figure}

In order to prove that the class of strong embeddability is a proper
subclass of weak embeddability, we take a detour via string graphs. 

A graph~$G= (V,E)$ is a \emph{string graph} if there exists a set~$\mathcal R =
\{R(v)\mid v\in V \}$ of curves in the plane such that~$R(u) \cap R(v)\not=
\emptyset$ if and only if~$\{u,v\}\in E$.
Deciding whether a graph is a string graph is \NP-hard~\cite{sss-rnp-03}.
However, Schaefer and {\u S}tefankovi{\u c}~\cite{ss-d-04} showed that a graph is
\emph{no} string graph if it is constructed from a non-planar graph by
subdividing each edge at least once.
Together with the following lemma we can thus prove that the pair of partitions
shown in Fig.~\ref{sfg:weaknotstrong} is not strongly embeddable.  %

\begin{lemma}\label{lem:strongString}
  The block intersection graph of a strongly embeddable pair of partitions is a
  string graph.
\end{lemma}

\begin{proof}
  Let~$\Gamma$ be a strong embedding of a pair of
  partitions~$\{\mathcal P_0, \mathcal P_1\}$. Our goal is to
  construct a set~$\mathcal R = \{R(B) \mid B \in \mathcal P_0 \cup \mathcal
  P_1\}$ of curves in the plane, which correspond to the blocks
  in~$\mathcal P_0 \cup \mathcal P_1$ such that~$R(B) \cap R(B') \not=
  \emptyset$ if and only if~$B$ and~$B'$ share a common element.  This
  is equivalent to the assertion that the block intersection graph
  $G_s(\mathcal P_0,\mathcal P_1)$ is a string graph.  We
  construct~$\mathcal R$ along~$\Gamma$ as follows. First we delete the points
  in~$\Gamma(U)$.  Then we delete one point of the boundary of each
  block region that is not an intersection point.  This results in a
  set~$R$ of curves that correspond to the blocks in $\mathcal P_0
  \cup \mathcal P_1$, and since~$\Gamma$ is strong, two curves have a
  common point if and only if the corresponding blocks have at least
  one common element and the previous block regions were not nested
  in~$\Gamma$.  For blocks whose block regions are nested
  in~$\Gamma$, we replace the curve that represents the nested block
  by a curve that crosses the surrounding block curve. This finally
  yields the desired set of curves. \qed
\end{proof}

Now consider the pair of partitions in Fig.~\ref{sfg:weaknotstrong}.
The left side of Fig.~\ref{sfg:weaknotstrong} shows a weak embedding
while the right side shows the corresponding block intersection
graph, which is constructed from~$K_5$ by subdividing each edge
exactly once.  Consequently, since~$K_5$ is not planar, it is no string
graph (according to Schaefer and {\u S}tefankovi{\u c}~\cite{ss-d-04}).  Applying
Lemma~\ref{lem:strongString} finally proves that the pair of
partitions depicted in Fig.~\ref{sfg:weaknotstrong} is not strongly
embeddable.

\subsection{Embeddability and Hypergraph Planarity}
\label{sub:embeddability_and_hypergraph_planarity}

 The weak
embeddability class forms the basis of the hierarchy and contains
all pairs of partitions. The strong embeddability class and
the full embeddability class are characterized by the existence of a
planar support and the planarity of the bipartite map of a pair of
partitions, respectively, where the latter directly implies a linear
time algorithm for the corresponding decision problem.  Moreover,
these characterizations reveal close relations to the hypergraph planarity
concepts of \emph{Zykov} and \emph{vertex planarity}.

A hypergraph $H = (V,\mathcal{S})$ is Zykov-planar~\cite{z-h-74}, if
there exists a subdivision of the plane into faces, such that each
hyperedge $S \in \mathcal{S}$ can be mapped to a face of the
subdivision, and each vertex $v \in V$ can be mapped to a point on the
boundary of all faces that represent a hyperedge containing~$v$.
Walsh~\cite{w-hvbm-75} showed that a hypergraph is Zykov planar if and
only if its bipartite map is planar.  %

In contrast, a hypergraph $H = (V,\mathcal{S})$ is
vertex-planar~\cite{jp-hpcdvd-87} if there exists a subdivision of the
plane into faces, such that every vertex $v \in V$ can be mapped to a
face and for every hyperedge $S \in \mathcal{S}$, the interior of the
union of all faces of the vertices in $S$ is connected.  Kaufmann et
al.~\cite{kks-sdh-09} showed that a hypergraph is vertex planar if and
only if it has a planar support.  
This shows that the class of \emph{fully}
embeddable pairs of partitions is a subclass of Zykov planar
hypergraphs, and the class of \emph{strongly}
embeddable pairs of partitions is a subclass of vertex planar
hypergraphs.

\section{Complexity of Deciding Strong Embeddability}\label{s:cdse}

In this section we show the \NP-completeness of testing strong embeddability. As a consequence, testing whether the corresponding hypergraph of a pair of partitions has a planar support is also \NP-complete by Theorem~\ref{thm:se_vertex_planar}.
This seems not very surprising considering the more general hardness
results of Johnson and Pollak~\cite{jp-hpcdvd-87} and Buchin et
al.~\cite{bkms-psh-10} who showed that deciding the existence of a
planar support and a 2-outerplanar support in general hypergraphs is
\NP-hard.  However, we consider a
restricted subclass of 2-regular hypergraphs, thus, the \NP-hardness
of our problem does not directly follow from the previous results.
Moreover, other special cases, e.g., finding path, cycle, tree, and
cactus supports are known to be solvable in polynomial
time~\cite{jp-hpcdvd-87,bkms-psh-10,bcps-bhaho-11}. Together with the characterization of
Theorem~\ref{thm:se_vertex_planar}, Theorem~\ref{thm:StrongHard}
immediately implies that testing the vertex planarity of a 2-regular
hypergraph is \NP-complete.

\begin{theorem}\label{thm:StrongHard}
  Deciding the strong embeddability of a pair of partitions is \NP-complete.
\end{theorem}

The existing hardness results~\cite{jp-hpcdvd-87,bkms-psh-10} rely on
elements that are contained in more than two hyperedges and could not be
adapted to our 2-regular setting.  Instead we prove the hardness of deciding
strong embeddability by a quite different reduction from the
\NP-complete problem \textsc{monotone planar 3Sat}~\cite{bk-obspp-10}.
A monotone planar 3Sat formula~$\varphi$ is a 3Sat formula whose
clauses either contain only positive or only negated literals (we call
these clauses \emph{positive} and \emph{negative}) and whose
variable-clause graph $H_\varphi$ is planar.  A \emph{monotone
  rectilinear representation} (MRR) of $\varphi$ is a drawing of $H_\varphi$
such that the variables correspond to axis-aligned rectangles on the
x-axis and clauses correspond to non-crossing E-shaped ``combs'' above
the x-axis if they contain only positive variables and below the
x-axis otherwise; see
Fig.~\ref{fig:planar_monotone_3_sat}.

An instance of \textsc{monotone planar 3Sat} is an MRR of a monotone planar 3Sat formula~$\varphi$.  In the
proof of Theorem~\ref{thm:StrongHard} we will construct a pair of
partitions~$\{\mathcal P_0,\mathcal P_1\}_\varphi$ that admits a
strong embedding if and only if~$\varphi$ is satisfiable.

For the sake of simplicity, we restrict
the class of
strong embeddings to the subclass of \emph{proper strong embeddings}, which
is equivalent, as we can argue that a pair of
partitions has a strong embedding if and only if it also has a proper one.
A strong embedding is \emph{proper} if the contour
graph does not contain background or linking faces that are
adjacent to only two other faces. 
Figure~\ref{fig:proper} illustrates how background or linking faces violating this condition can be removed, transforming a strong embedding into a proper one. %
\begin{figure}[tbp]
  \centering \subfigure[Removing a linking
  face.\label{fig:proper-linking}]{\makebox[.4\textwidth]{\includegraphics[page=1]{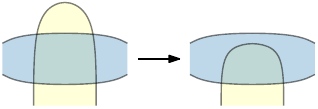}}}
  \hspace{1cm} \subfigure[Removing a background
  face.\label{fig:proper-bg}]{\makebox[.4\textwidth]{\includegraphics[page=2]{proper}}}
  \caption{Two cases for transforming a strong embedding into a proper
    strong embedding.}
	\label{fig:proper}
\end{figure}
We say that two proper strong embeddings are \emph{equivalent} if the
embeddings of their contour graphs are equivalent, i.e. if the cyclic order of the edges around each vertex is the same. A pair of
partitions has a \emph{unique strong embedding} if all proper strong
embeddings are equivalent. 
Note that, analogously to the definition of equivalence of planar graph embeddings, two equivalent proper strong embeddings may have different unbounded outer background faces.
Our construction in the hardness proof is independent of the choice of the outer face.

Next we define a special pair of partitions that has a unique
grid-shaped embedding as a scaffold for the gadgets in the subsequent
proof of Theorem~\ref{thm:StrongHard}. The first step is to construct
a base graph $G_{m,n}$ for two integers~$m$ and~$n$. The graph
$G_{m,n}$ is a grid with $mn+1$ columns and $2m+2$ rows of
vertices with integer coordinates $(i,j)$ for $0 \le i \le mn$ and $0
\le j \le 2m+1$. Each vertex $v$ with coordinates $(i,j)$ is connected
to the four vertices at coordinates $(i-1,j), (i+1,j), (i, j-1),
(i,j+1)$ (if they exist). Between the middle rows $m$ and $m+1$ we
remove all 
vertical edges except for those in
columns $0, m, 2m, \dots, nm$. This defines
$n$ larger grid cells of width~$m$ in this particular row. Figure~\ref{fig:grid}
(left) shows an example. %

\begin{figure}[tbp]
	\centering
		\includegraphics[scale=1]{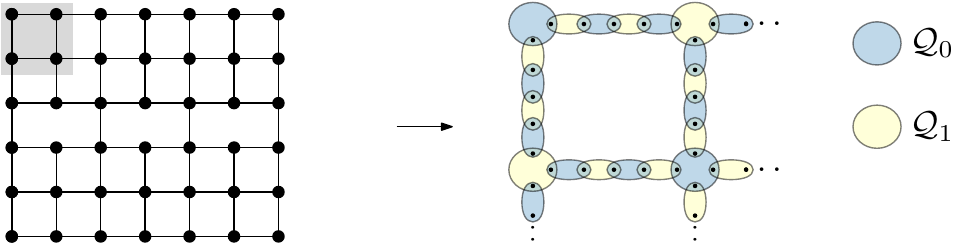}
                \caption{Graph $G_{2,3}$ and the partitions $\{\mathcal Q_0,
                  \mathcal Q_1\}$ sketched for the top-left grid cell marked in gray.}
	\label{fig:grid}
\end{figure}

From $G_{m,n}$ we construct a pair of partitions $\{\mathcal Q_0,
\mathcal Q_1\}$ as follows (see Fig.~\ref{fig:grid}). For each
vertex~$v$ with coordinates $(i,j)$ we create a \emph{vertex
block} $B_v$ in partition $\mathcal
Q_{(i+j) \pmod{2}}$. For each edge $(u,v)$ in $G_{m,n}$ we create a
chain of four \emph{edge blocks} $B_{u,v}^1$, $B_{u,v}^2$, $B_{u,v}^3$,
$B_{u,v}^4$, such that $B_{u,v}^1$ and $B_{u,v}^3$ are in the same 
partition as $B_v$ and $B_{u,v}^2$ and $B_{u,v}^4$ are in the same
partition as $B_u$.
We distribute five distinct elements among the edge blocks of $(u,v)$
and the vertex blocks for $u$ and $v$ such that they form the desired chain pattern and each intersection face contains one common element. %
The
pair $\{\mathcal Q_0, \mathcal Q_1\}$ is indeed a pair of partitions
as every element belongs to
exactly one block of each partition. Edge blocks contain two and
vertex blocks up to four elements (depending on the degree of the
corresponding vertex in $G_{m,n}$). Below we will add the gadgets of the reduction on top of $\{\mathcal Q_0, \mathcal Q_1\}$, for which it is required that there is an edge block in each partition that does not share any element with a vertex block. This explains why we link  blocks of adjacent vertices by chains of four blocks. %

The next lemma shows that $\{\mathcal Q_0, \mathcal Q_1\}$ has a unique embedding, which is a consequence of the fact that %
$G_{m,n}$ is a
  subdivision of a planar 3-connected graph (assuming $n\ge 2$) and thus it has a unique embedding. %
This property is inherited by $\{\mathcal Q_0, \mathcal Q_1\}$ in our
  construction. %

\begin{lemma}\label{lem:unique_embed}
  The pair of partitions $\{\mathcal Q_0, \mathcal Q_1\}$ has a
  unique embedding.
\end{lemma}

\begin{proof}
  First, we observe that the base graph $G_{m,n}$ is a
  subdivision of a planar 3-connected graph (assuming $n \ge 2$) and thus it has a unique embedding (up to the
  choice of the outer face) in the plane.  We claim that this property
  is inherited by $\{\mathcal Q_0, \mathcal Q_1\}$ in our
  construction.
	
  Each edge block contains exactly two elements and intersects
  exactly two blocks of the other partition.  Thus its contour
  subgraph in any proper strong embedding is isomorphic to the
  4-cycle $C_4$ with two non-incident duplicate edges inside, which
  belong to the boundaries of the two intersecting blocks.  Each
  vertex block $B_v$ contains two, three, or four elements,
  depending on the degree $k$ of $v$ in $G_{m,n}$.  Since $B_v$
  intersects with $k$ edge blocks of the other partition there are
  exactly two intersection points with the boundary of each of these
  edge blocks in a proper strong embedding (if there were four
  intersection points, then the edge block would not be proper).
  Thus the contour subgraph of $B_v$ is a 4-, 6-, or 8-cycle with
  two, three, or four non-incident duplicate edges inside belonging to
  the intersecting edge blocks.  There is a bijection between the
  possible cyclic intersection orders of the $k$ edge blocks and the
  possible cyclic orders of the $k$ incident edges of vertex $v$ in
  $G_{m,n}$.  Thus we have locally the same embedding choices of
  the contour graph of $C_v$ as for the vertex $v$ in $G_{m,n}$.
  Since $G_{m,n}$ has a unique embedding, and since each edge of
  $G_{m,n}$ is represented in $\{\mathcal Q_0, \mathcal Q_1\}$
  by a sequence of four edge blocks with a locally unique embedding
  between the two incident vertex blocks, we conclude that for every
  proper strong embedding of $\{\mathcal Q_0, \mathcal Q_1\}$ the
  induced contour graph is the same graph with the same unique planar
  embedding. Otherwise we could derive two different embeddings of
  $G_{m,n}$, which is a contradiction. \qed
\end{proof}

Now we have all the tools that we need to prove our main theorem in this section.

\begin{proof}[of Theorem~\ref{thm:StrongHard}]
  First we show that the problem is in \NP. By
  Theorem~\ref{thm:se_vertex_planar} we know that a pair of
  partitions is strongly embeddable if and only if it has a planar
  support.  Thus we can ``guess'' a graph on $U$ and then test
  its planarity and support property in polynomial time. This shows membership in \NP.
  It remains to describe the hardness reduction.

  Let $\varphi$ be a planar monotone 3Sat formula together with an
  MRR.  %
First we construct the pair of
  partitions $\{\mathcal Q_0, \mathcal Q_1\}$ for the base graph
  $G_{m,n}$, %
  where
  $m$ is the number of clauses of $\varphi$ and $n$ is the number of variables
of $\varphi$.
By Lemma~\ref{lem:unique_embed}
  $\{\mathcal Q_0, \mathcal Q_1\}$ has a unique proper grid-like
  embedding.  We call $\{\mathcal Q_0, \mathcal Q_1\}$ the \emph{base
    grid} and the $n$ special cells between rows $m$ and $m+1$ the
  \emph{variable cells} of the base grid.
	
		\begin{figure}[tb]%
			\begin{center}
                          \subfigure[Monotone rectilinear
                          representation of a formula $\varphi$]{
                            \def\svgwidth{0.75\textwidth}
                            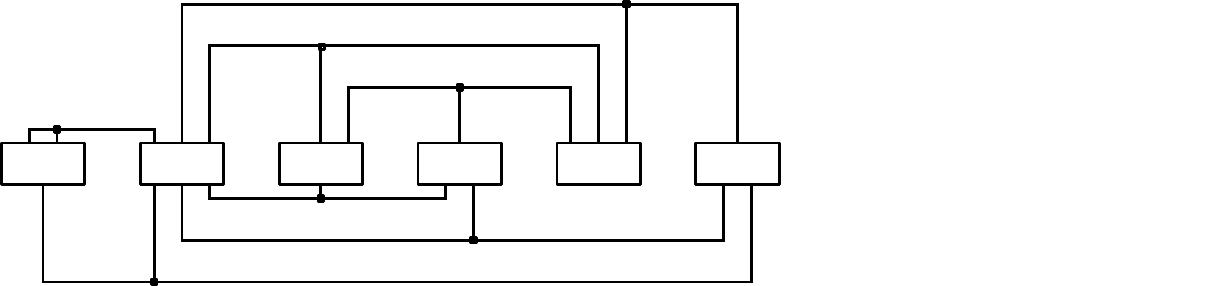
					\label{fig:planar_monotone_3_sat}%
                                      } \subfigure[Sketch of the
                                      clause blocks laid on top of
                                      the grid $\{\mathcal Q_0,
                                      \mathcal Q_1\}$ (empty columns omitted)
									  ]{\makebox[.98\textwidth][c]{
                                        \def\svgwidth{0.8\textwidth}
                                        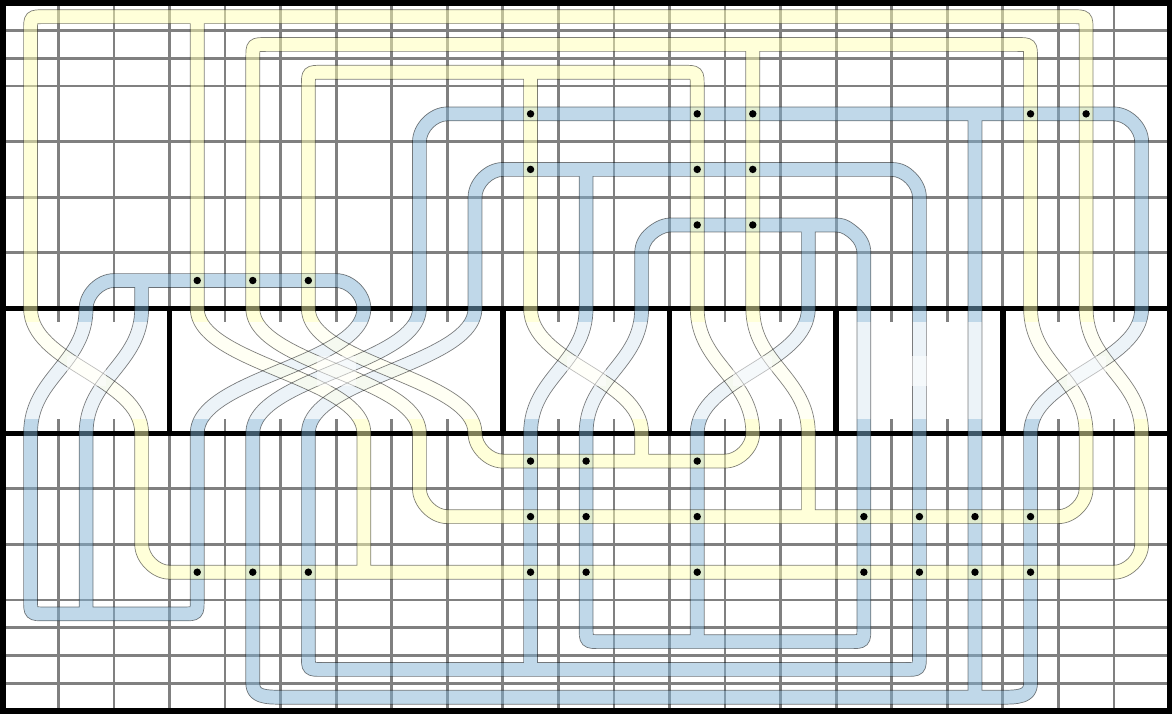}
					\label{fig:np_complete_2_clusterings}%

				}
				\caption{Illustration of the \NP-hardness reduction}
			\end{center}
		\end{figure}
		
                Next we augment the pair of partitions $\{\mathcal
                Q_0, \mathcal Q_1\}$ by additional blocks, one for
                each clause, where positive clauses are added to
                $\mathcal Q_0$ and negative clauses to $\mathcal Q_1$.
                The definition of these \emph{clause blocks}
                closely follows the layout of the
                given MRR, see
                Fig.~\ref{fig:planar_monotone_3_sat}. Let $C_1, C_2,
                \dots, C_l$ be the positive clauses of $\varphi$
                ordered so that if $C_i$ is nested inside the E-shape
                of $C_j$ in the given MRR then $i<j$. Analogously let
                $C_{l+1}, \dots, C_m$ be the ordered negative
                clauses. We describe the definition of the block $B_i$
                for a positive clause $C_i$ ($1 \le i \le l$); blocks
                for negative clauses are defined symmetrically. We
                create an \emph{intermediate embedding} of $B_i$
                (which is not yet strong but serves as a template for
                a later strong embedding) by putting $B_i$ on top of
                the base grid\footnote{The idea of fixing paths to an
                  underlying grid is inspired by Chaplick et
                  al.~\cite{cjkv-bbpigsnwg-12}.} and adding new
                elements to $B_i$ and to certain edge blocks in
                $\mathcal Q_1$. This fixes $B_i$ to run through two
                mirrored E-shaped sets of grid cells of our choice
                (Fig.~\ref{fig:np_complete_2_clusterings}). %
                In the upper half of the base grid, $B_i$ is assigned
                to run between rows $m-i$ and $m-i+1$. Furthermore,
                $B_i$ is assigned to three columns leading towards the
                variable cells from the top. Let $x_j$ be a variable
                contained in $C_i$ and assume that $C_i$ is the $k$-th
                positive clause from the right connecting to $x_j$ in
                the embedding of the given MRR. Then $B_i$ runs
                between columns $jm - k$ and $jm - k + 1$. In the
                lower half of the base grid we translate and mirror
                the resulting E-shape as follows. We let $B_i$ occupy
                the cells between rows $2m+2-l+i-1$ and $2m+2-l+i$ and
                the three columns are shifted to the left by the
                number of occurrences of the respective variable in
                negative clauses
                (Fig.~\ref{fig:np_complete_2_clusterings}). Since each
                variable cell is $m$ columns wide, we can always
                assign each clause to a unique column of $x_j$ in the
                top and bottom half of the grid in this way. %

				We actually
                fix $B_i$ to the base grid by adding one shared element for each crossed edge of a grid cell to both
                $B_i$ and the respective edge block of $\mathcal Q_1$ that does not share an element with a vertex
                block in $\mathcal Q_0$ (recall that $\{\mathcal Q_0, \mathcal Q_1\}$ contains such a
                block in each partition and for each grid edge). No two blocks of the same clause type (positive or negative)
                intersect, but blocks of different type do intersect in certain grid cells. For each grid cell shared
                between a positive and negative block (except for the $n$ variable cells) we add one shared element
                (black dots in Fig.~\ref{fig:np_complete_2_clusterings}) and call the respective grid cell the
                \emph{home cell} for this element. Recall that the orders of the incoming blocks from the top
                and the bottom of each variable cell are inverted. Thus, within each variable cell the blocks of each
                pair of a positive and negative clause using the corresponding variable intersect, but no shared
                element is added. We denote the resulting new pair of partitions as~$\{\mathcal P_0,\mathcal
                P_1\}_\varphi$ and observe that its size is polynomial in the size of $\varphi$.
	
                Next we argue about the strong embedding options in contrast to the immediate embedding for a clause
                block $B_i$ in $\{\mathcal P_0, \mathcal P_1\}_\varphi$. In the intermediate embedding each block has three
                connections through variable cells linking the upper E-shape with the lower E-shape. Any element shared
                with an edge block of the uniquely embedded base grid must obviously be reached by the block region of
                $B_i$. Since the block region must be simple, any strong embedding of $B_i$ results from opening the
                intermediate embedding of $B_i$ in exactly two grid cells so that the resulting block region of $B_i$
                is connected and has no holes. Additionally, a shared element must be placed in any intersection of the
                block region of $B_i$ with block regions of other clause blocks.
	
                First we assume that $\varphi$ is a satisfiable formula and a satisfying variable assignment is given.
                We need to show that $\{\mathcal P_0,\mathcal P_1\}_\varphi$ has a strong embedding. If a variable
                $x_j$ has the value \emph{true} in the given assignment we open all blocks of negative clauses using
                $x_j$ in the corresponding variable cell; if $x_j$ is \emph{false} we open all blocks of positive
                clauses using~$x_j$. Thus no blocks intersect in variable cells any more. If a clause contains more
                than one \emph{true} literal, we open all but one connection in its variable cells of \emph{true}
                literals. Since the assignment satisfies $\varphi$, we know that each clause block is opened exactly
                twice in its variable cells and thus forms a valid simple block region. Moreover, we place all shared
                elements in their home cells so that every block intersection contains an element and the
                embedding is strong. We call a strong embedding of $\{\mathcal P_0,\mathcal P_1\}_\varphi$ with the
                above properties a \emph{canonical embedding}.
	
                Now assume that $\{\mathcal P_0,\mathcal P_1\}_\varphi$ has a strong embedding. We know that the base
                grid has its unique embedding and that each block is embedded as a simple region that results from
                opening the intermediate embedding (with its two E-shapes linked through three variable cells) in
                exactly two cells. If the embedding is already canonical, we can immediately construct a satisfying
                variable assignment for $\varphi$: if a variable cell is crossed by clause blocks in $\mathcal Q_0$ we
                set the variable to \emph{true}, otherwise we set it to \emph{false}. Since every clause block is
                connected we know that this assignment satisfies all clauses. If the embedding is not canonical we show
                that it can be transformed into a canonical embedding as follows. In a non-canonical embedding it is
                possible that two blocks $B_i$ and $B_j$ intersect in a variable cell $x_k$ and have a shared element
                in their intersection face in the cell of $x_k$ rather than in the home cell of that element. This
                means, on the other hand, that in some shared home cell $\gamma$ of $B_i$ and $B_j$, say in the
                upper half, at least one of the two blocks is opened (as there is no more shared element to put into an
                intersection face). Thus the grid cell $\gamma$ splits the E-shaped block region of one or both blocks
                in the upper half into two disconnected components, meaning that each opened block crosses at least two
                variable cells in order to connect both components via the lower half. Hence we can safely split any
                block that is opened in $\gamma$ in the cell of variable $x_k$, re-connect it inside $\gamma$, and
                place the shared element of $B_i$ and $B_j$ into its home cell $\gamma$. This removes the block
                intersection in the cell of $x_k$. Once all block crossings within variable cells are removed, the
                resulting embedding is a canonical embedding and we can derive the corresponding satisfying variable
                assignment. \qed
\end{proof}

\section{Extensions and Conclusion} %
\label{sec:further_extensions_and_conclusion}

We have characterized three main embeddability
classes for pairs of partitions, which in fact form a strict
hierarchy, and we have shown \NP-completeness of
deciding strong embeddability.  
From a practical point of view the class of strong embeddings is of particular interest: it guarantees that every intersection between block regions is meaningful as it contains at least one element, but on the other hand allows blocks to cross, imposing less restrictions than full embeddings.
Interesting
subclasses of strong embeddings that further structure the space
between strong and full embeddability can be defined and we mention two of them.

\begin{figure}[htb]%
	\centering
		\subfigure[]{
			\includegraphics[width=0.49\textwidth]{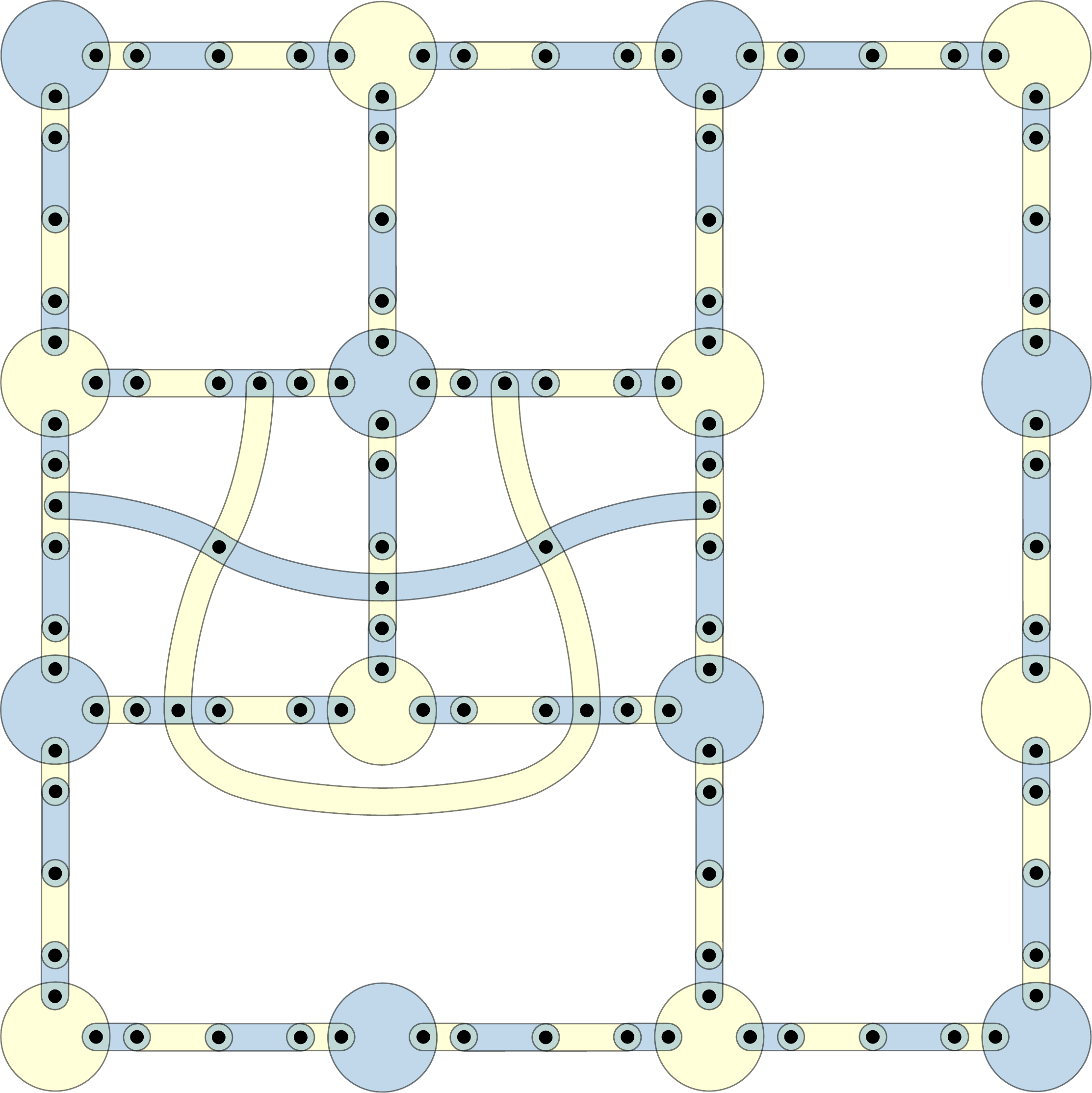}			
			\label{fig:2_clusterings_no_single_intersection_detailed}%
		}
		\begin{minipage}[b]{0.49\textwidth}
			\centering
		\subfigure[]{
			\includegraphics[width=0.25\textwidth]{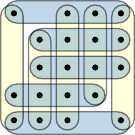}			
			\label{fig:plane_grid_strongly_embeddable}%
		}	\\
		\subfigure[]{
			\includegraphics[width=\textwidth]{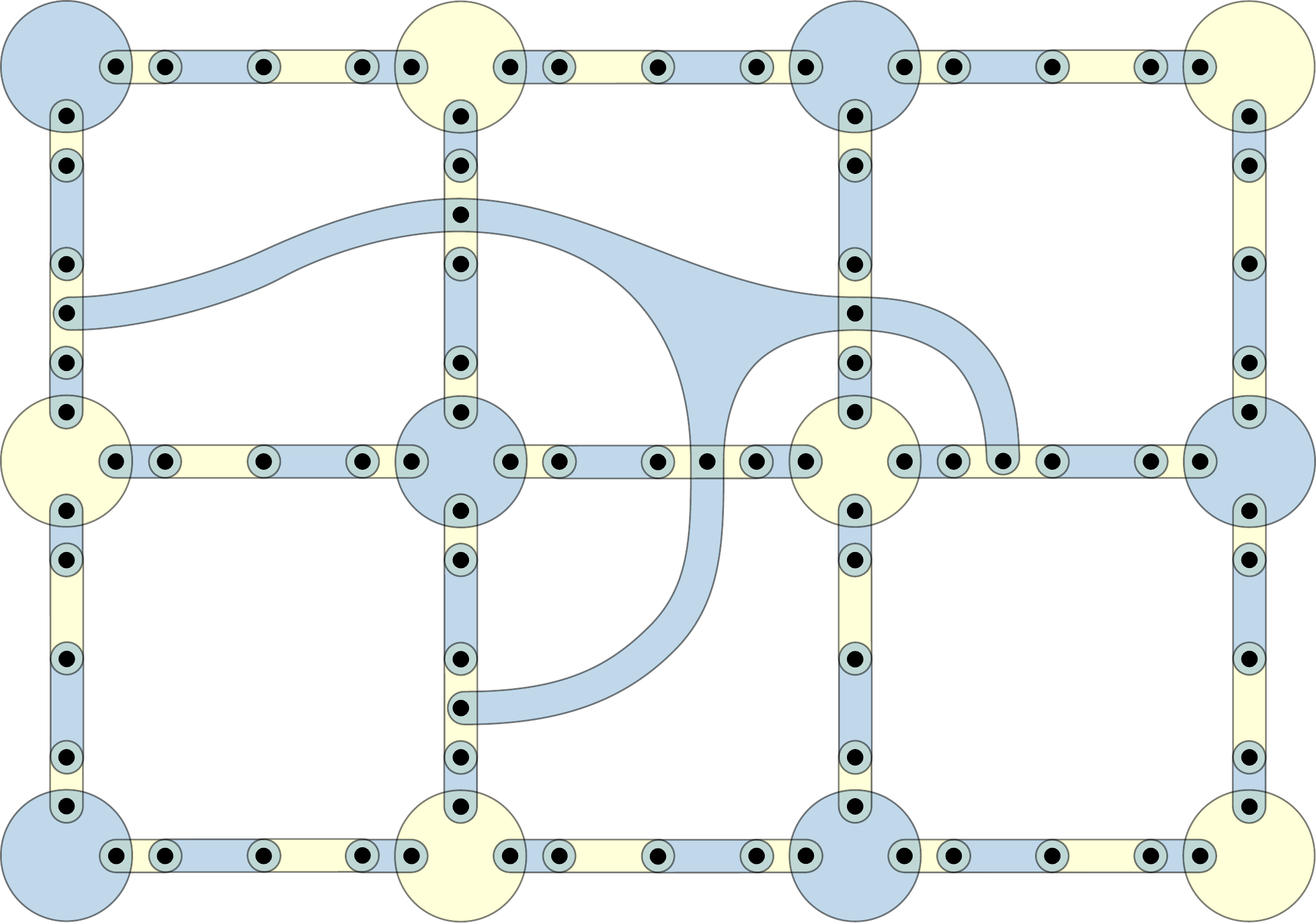}			
			\label{fig:2_clusterings_no_path_based_detailed}%
		}		
		\end{minipage}

		\caption{(a) Two partitions based on a uniquely embedded grid that have a strong embedding, but no single-intersection strong embedding. (b) A strong grid embedding. (c) Two partitions based on a uniquely embedded grid that have a single-intersection strong embedding but no strong grid embedding.}%
		\label{fig:2_clusterings_grid_au
gmentations}%

\end{figure}

In a strong embedding two blocks can intersect many times forming
disjoint intersection faces, whereas a full embedding permits only a
single connected intersection region for any pair of blocks (this is
also a common requirement for Euler diagrams~\cite{c-gdaevd-07}).  In
\emph{single-intersection strong embeddings} we adapt this
\emph{unique intersection region} condition of full embeddings, but
still permit that two blocks cross in the embedding.  This new
class is a true subclass of strong embeddings, see the example in
Fig.~\ref{fig:2_clusterings_no_single_intersection_detailed}.  It is
open  whether the corresponding decision problem is still
\NP-complete since our proof is based on the existence of multiple
intersection regions between pairs of blocks.

Another interesting subclass are \emph{strong grid embeddings}, in
which the blocks of $\mathcal P_0$ and $ \mathcal P_1$ are embedded as horizontal and vertical ribbons, respectively,
which intersect in a matrix-like fashion, see
Fig.~\ref{fig:plane_grid_strongly_embeddable}.  Obviously a strong
grid embedding is also a single intersection strong embedding, but
again strong grid embeddings form a true subclass, see
Fig.~\ref{fig:2_clusterings_no_path_based_detailed}.  It is easy to see that a pair of
partitions admits a strong grid embedding if and only if its block
intersection graph is a \emph{grid intersection
  graph}~\cite{hnz-gig-91}, i.e., an intersection graph of horizontal
and vertical segments in the plane.
Kratochv\'{i}l~\cite{k-spspcnpc-94} showed that
deciding whether a bipartite graph is a grid-intersection graph is
\NP-complete.  
All fully
embeddable pairs of partitions have a strong grid embedding: The
bipartite map of a fully embeddable pair of partitions is planar by
Theorem~\ref{thm:charfull} and immediately induces a planar bipartite
block intersection graph, which, according to Hartman et
al.~\cite{hnz-gig-91}, is a grid intersection graph.  This implies \NP-completeness of deciding strong grid embeddability. But as the
example of Fig.~\ref{sfg:strongnotfull} shows, not every instance with
a strong grid embedding admits a full embedding.

It is an interesting direction for future work to study the generalization of our embeddability concepts to $k > 2$ partitions. While weak embeddability and its properties extend readily to any number of partitions, it is less obvious how to generalize strong and full embeddability. One possibility is to require the properties in a pairwise sense; otherwise constraints for new types of faces in the contour graph belonging to more than one but less than $k$ block regions might be necessary. %

On the practical side, future work could be the design of algorithms
that find visually appealing simultaneous embeddings according to our
different embeddability classes.  Finally, if the partitions are
clusterings on a graph, one would ideally want to simultaneously draw
both the partitions and the underlying graphs.

\end{document}